\documentclass{lmcs} %%% last changed 2014-08-20

\usepackage{hyperref}
\theoremstyle{plain} %\crefname{satz}{Satz}{S\"atze}
\DeclareMathOperator*{\E}{E}

\begin{document}

% If the title is longer than 55 characters, then specify a shorter running title as the optional argument to \title. The running title should be roughyl at most 55 characters:
\title{Feasible disjunction for random resolution}

% If you are not submitting to a special issue, you may remove the line below.
% If you are submitting to a special issue, you do not need to fill out this
% field (although you may).  Please be aware, however, that upon publication
% the layout editing team will insert the full special issue name, which may
% push down the text of your first section.
%\specialissue{Instructions for \LaTeX{} formatting of LMCS articles}

% affiliations are numbered automatically with a, b, c (see below)
% use the optional argument to indicate the affiliation(s) of each author
% omit the argument if there is only one author, or only one affiliation
\author{Theodoros Papamakarios\lmcsorcid{0009-0009-2814-5256}}
\thanks{Part of this work was done while the author was at the Institute of Mathematics of the Czech Academy of Sciences and supported by the GA{\v C}R grant 23-04825S}

%% etc.

%% required for running head on odd and even pages, use suitable
%% abbreviations in case of long titles and many authors:

%%%%%%%%%%%%%%%%%%%%%%%%%%%%%%%%%%%%%%%%%%%%%%%%%%%%%%%%%%%%%%%%%%%%%%%%%%%

%% the abstract has to PRECEDE the command \maketitle:
%% be sure not to issue the \maketitle command twice!

\begin{abstract}
We show that a (stronger) version of random resolution has the feasible disjunction property.  This is the first instance of a proof system not known to have feasible interpolation, which nevertheless has the feasible disjunction property.
\end{abstract}

\maketitle

%% start the paper here:
\section{Introduction}

We say that a proof system $P$ has the \emph{feasible interpolation} property if
\begin{quote}
\medskip
given a $P$ refutation of a formula of the form $F\wedge G$, where $F$ and $G$ do not share any variables,  we can decide in polynomial time which of $F$ and $G$ is unsatisfiable.
\medskip
\end{quote}
Feasible interpolation is a central object of study in proof complexity.  In particular,  if $P$ has the feasible interpolation property,  then we get conditional (or unconditional in the case of monotone feasible interpolation) lower bounds for $P$ (see e.g.\ \cite{Kraj97,Pudl97}).  If on the other hand $P$ does not have the feasible interpolation property,  this implies that $P$ is not automatable,  that is,  short $P$ refutations are hard to find (see e.g.\ \cite{Boneetal00,Becketal14}).

A closely related but less studied concept is \emph{feasible disjunction} \cite{Rudi97,KrajPudl98,Pudl03}.  We say that a proof system $P$ has the feasible disjunction property if
\begin{quote}
\medskip
whenever a formula of the form $F\wedge G$,  where $F$ and $G$ do not share any variables,  has a $P$ refutation of size $s$,  then either $F$ or $G$ must have a $P$ refutation of size polynomial in $s$.
\medskip
\end{quote}

The two properties are related in that typically,  systems which have feasible interpolation also have the feasible disjunction property.  In fact,  in all known cases,  the arguments themselves showing feasible interpolation can be straightforwardly modified to not only decide which formula is unsatisfiable but also produce a refutation of it.  Moreover,  in some cases,  feasible interpolation arguments go through feasible disjunction (see e.g.\ \cite{Hako20}).  In the other direction,  it is shown in \cite{Garl24} that feasible disjunction fails for Res(2),  which is precisely the point in the bounded-depth Frege hierarchy where feasible interpolation seems to fail too.

Having said that,  the two properties are not equivalent.  In one direction,  feasible interpolation is downward closed under p-simulation (that is,  if $P$ has feasible interpolation and $P$ p-simulates $Q$, then $Q$ has feasible interpolation),  whereas this is not true for feasible disjunction.  Indeed,  it was recently shown \cite{Hubaetal24,Lietal24} that by restricting a proof system with both feasible interpolation and feasible disjunction so that it corresponds to an intersection of TFNP classes,  one gets a subsystem for which feasible disjunction fails.

In this paper we give the first example in the opposite direction,  in the following weaker sense: we identify a proof system,  namely an extension of random resolution,  which has the feasible disjunction property,  yet is not known to have feasible interpolation.  Our proof is based on the minimax theorem,  and differs from previous arguments in an essential way.  Typically,  feasible disjunction is shown by directly extracting a refutation of either $F$ or $G$ from the given refutation of $F\wedge G$,  as is done for instance for resolution and cutting planes \cite{Pudl97}. For algebraic or semi-algebraic proof systems,  one can instead argue indirectly via duality (see \cite{Hako20,Hubaetal24}),  but the semantic object dual to the refutation lives in a space of dimension polynomial in it,  so a refutation can still be found in polynomial time by search \cite{Hako20}.  Our argument produces neither: the object we obtain is a distribution over assignments,  of exponential support; in particular,  we do not see how to make the proof polynomial-time constructive.

\section{Random resolution refutations}

Let $F$ be a CNF formula and let $V=V(F)$ be the set of variables occurring in $F$.  A \emph{resolution refutation} of $F$ is a sequence of clauses ending with the empty clause and such that each clause in the sequence is either a clause of $F$ or results from earlier clauses by either the \emph{resolution rule} ``from $C\vee x$ and $D\vee \neg x$ derive $C\vee D$'' or the \emph{weakening rule} ``from $C$ derive $C\vee D$''.  Now, let $\Pi$ be a resolution refutation of $F$ and let $\alpha:I\to\{0,1\}$ be a truth assignment where $I\subseteq V$.  Applying $\alpha$ to every clause of $\Pi$ and deleting the clauses that get satisfied, we get a resolution refutation of the restricted formula $F\vert_\alpha$,  which we denote by $\Pi\vert_\alpha$.  Finally,  for a truth assignment $\alpha:V\to\{0,1\}$,  we define \emph{the traversal of $\Pi$ by $\alpha$},  $\tau_{\scriptscriptstyle\Pi}(\alpha)$,  to be the clause of $F$ we reach if we traverse $\Pi$ starting at the empty clause and always choosing at applications of derivation rules the premise which is falsified by $\alpha$.  Traversals and restrictions by partial assignments interact as follows.
\begin{lem}\label{lemm:trav}
Let $I\subseteq V$ and let $\alpha:I\to\{0,1\}$ and $\beta:V\setminus I\to\{0,1\}$ be truth assignments.  We have that
\[
\tau_{\scriptscriptstyle\Pi\vert_\beta}(\alpha)=\tau_{\scriptscriptstyle\Pi}(\alpha,\beta)\vert_\beta.
\]
\end{lem}
\begin{proof}
Clauses satisfied by $\beta$ are deleted in $\Pi\vert_\beta$,  so the two traversals must follow the same path.
\end{proof}

Random resolution was introduced in \cite{Bussetal14} and further studied in \cite{PudlThap19}.  It is a proof system which models the notion of resolution proofs that make mistakes.  The papers \cite{Bussetal14,PudlThap19} consider one type of error of random resolution refutations which we call static error.  For the purposes of showing feasible disjunction,  we  introduce a different error measure based on the notion of traversals,  which we will call dynamic error.

\begin{defi}
A \emph{random resolution refutation} of $F$ is a probability distribution $(B_i,\Pi_i)_{i\sim\Delta}$,  where  
$B_i$ is a CNF with variables from $V=V(F)$ and $\Pi_i$ is a resolution refutation of $F\wedge B_i$.  For an $\alpha:V\to\{0,1\}$,  the \emph{static error} of $(B_i,\Pi_i)_{i\sim\Delta}$ on $\alpha$ is defined as the probability $\Pr_{i\sim\Delta}\left[B_i\vert_\alpha=0\right]$.  The \emph{dynamic error} of $(B_i,\Pi_i)_{i\sim\Delta}$ on $\alpha$ is defined as the probability $\Pr_{i\sim\Delta}\left[\tau_{\scriptscriptstyle\Pi_i}(\alpha)\in B_i\right]$.  The static (dynamic) error of $(B_i,\Pi_i)_{i\sim\Delta}$ is defined as the maximum over all $\alpha:V\to\{0,1\}$ of the static (resp.\ dynamic) error of $(B_i,\Pi_i)_{i\sim\Delta}$ on $\alpha$.
\end{defi}

An intuitive way to understand a random resolution refutation of $F$ is as a probabilistic procedure solving the search task $\text{Search}(F)$ associated with $F$: given an $\alpha:V\to\{0,1\}$,  find a clause of $F$ falsified by $\alpha$.  Namely,  given an $\alpha$,  we choose $i$ according to $\Delta$,  and carry out the traversal of $\Pi_i$ by $\alpha$ to find a falsified clause.  Dynamic error is the failure probability of this procedure.  Static error is instead a property of the formulas $B_i$ alone---it measures how often they are true---and does not depend on the refutations $\Pi_i$.  Both error measures give rise to sound proof systems,  which moreover are robust in the sense that they turn out to be invariant under various changes.  We mention here two such results,  namely \cite[Lemma 2.2]{PudlThap19} and \cite[Lemma 2.3]{PudlThap19}, relevant to this paper; we refer to \cite{PudlThap19} for more.
\begin{prop}\label{prop:prop}
Suppose $F$ has a random resolution refutation $(B_i,\Pi_i)_{i\sim\Delta}$ of static (dynamic) error $\varepsilon\in(0,1)$.  Then
\begin{enumerate}
\item $F$ is unsatisfiable;
\item $F$ has a random resolution refutation $(B_i,\Pi_i)_{i\sim\Delta'}$ of static (resp.\ dynamic) error $2\varepsilon$,  such that $\Delta'$ is the uniform distribution on a sample space of size $O(n/\varepsilon)$,  where $n$ is the number of variables of $F$;
\item $F$ has a random resolution refutation $(B_I,\Pi_I)_{I\sim\Delta^k}$ of static (resp.\ dynamic) error $\varepsilon^k$,  such that for every $I=(i_1,\dotsc,i_k)$ the size of $\Pi_I$ is $O(s_1\dotsb s_k)$,  where $s_i$ is the size of $\Pi_i$.
\end{enumerate}
\end{prop}
\begin{proof}
The constructions showing 2 and 3 in \cite{PudlThap19} carry over to the case of dynamic error with no changes.  For soundness in the case of dynamic error,  suppose for the sake of contradiction that $F$ is satisfiable and let $\alpha$ be a satisfying assignment.  Then $\Pr_{i\sim\Delta}[\tau_{\scriptscriptstyle\Pi_i}(\alpha)\in B_i]=1$,  contradicting the fact that $\varepsilon<1$.
\end{proof}

Notice that for every random resolution refutation,  its dynamic error is at most its static error,  and therefore random resolution with static error $\varepsilon$ is a weaker system than random resolution with dynamic error $\varepsilon$.  In particular,  if the latter system has the feasible interpolation property,  then the former has it as well.

The \emph{size} of a random resolution refutation $(B_i,\Pi_i)_{i\sim\Delta}$ is the maximum size of the refutations $\Pi_i$.  We can assume,  by Proposition~\ref{prop:prop},  without loss of generality that a random resolution refutation is a finite object and its size is well defined.

We end this section by briefly summarizing what is known about interpolation for random resolution with static error.  Given a random resolution refutation $(B_i,\Pi_i)_{i\sim\Delta}$ of a formula $F(x,u)\wedge G(x,v)$,  it is straightforward to construct a DAG-like randomized communication protocol for the disjointness of the sets $A=\{\alpha: \text{$F(\alpha,u)$ is satisfiable}\}$ and $B=\{\beta: \text{$G(\beta,v)$ is satisfiable}\}$ (that is,  a protocol where Alice gets an $\alpha\in A$,  Bob gets a $\beta\in B$ and the two players find a coordinate $j$ for which $\alpha_j\neq\beta_j$).  Based on this, Pudl{\'{a}}k and Thapen \cite{PudlThap19} get an exponential lower bound for tree-like random resolution refutations from a lower bound for tree-like randomized communication protocols.  Kraj{\'i}{\v c}ek \cite{Kraj17} shows that one can get a circuit that separates $A'\subseteq A$ and $B'\subseteq B$,  where the sizes of $A'$ and $B'$ depend on the static error $\varepsilon$.  Provided that $\varepsilon\leq 1/d^2$, where $d=\max_i|B_i|$,  this yields an exponential lower bound for the clique-coloring formulas.  Finally, Kraj{\'i}{\v c}ek \cite{Kraj18} considers a circuit model separating $A$ and $B$,  where rectangles in which the protocol errs are replaced by oracles separating these rectangles.  However,  none of the above models gives a polynomial time interpolating algorithm.

\section{Feasible disjunction for random resolution with dynamic error}

Let $F$ and $G$ be two CNF formulas not sharing any variables such that $F\wedge G$ is unsatisfiable, and let $(B_i,\Pi_i)_{i\sim\Delta}$ be a random resolution refutation of $F\wedge G$ of dynamic error $\varepsilon$.  For a truth assignment $\alpha$ to the variables of $F$ and a truth assignment $\beta$ to the variables of $G$,  we set
\begin{align*}
f(\alpha,\beta)&\coloneqq\Pr_{i\sim\Delta}\left[\tau_{\scriptscriptstyle\Pi_i}(\alpha,\beta)\in F\right],\\
g(\alpha,\beta)&\coloneqq\Pr_{i\sim\Delta}\left[\tau_{\scriptscriptstyle\Pi_i}(\alpha,\beta)\in G\right].
\end{align*}
Notice that for all $\alpha,\beta$,  $f(\alpha,\beta)+g(\alpha,\beta)\geq 1-\varepsilon$.  Now,  we define
\begin{align*}
V_F&\coloneqq\max_\mu\min_\alpha \E_{\beta\sim\mu}\left[f(\alpha,\beta)\right],\\
V_G&\coloneqq\max_\nu\min_\beta \E_{\alpha\sim\nu}[g(\alpha,\beta)],
\end{align*}
where $\mu$ ranges over all probability distributions on assignments $\beta$ and $\nu$ ranges over all probability distributions on assignments $\alpha$.

\begin{lem}\label{lemm:main}
We have that
$V_F+V_G\geq 1-\varepsilon.$
\end{lem}
\begin{proof}
By von Neumann's minimax theorem,  we get \[V_F=\min_\nu\max_\beta \E_{\alpha\sim\nu}\left[f(\alpha,\beta)\right].\] Fix $\nu_0$ attaining the minimum.  Since $f(\alpha,\beta)+g(\alpha,\beta)\geq 1-\varepsilon$ for every $\alpha,\beta$,  we get that $\E_{\alpha\sim\nu_0}\left[f(\alpha,\beta)\right]+\E_{\alpha\sim\nu_0}\left[g(\alpha,\beta)\right]\geq 1-\varepsilon$ for every $\beta$,  in particular \[\E_{\alpha\sim\nu_0}\left[f(\alpha,\beta_0)\right]+\E_{\alpha\sim\nu_0}\left[g(\alpha,\beta_0)\right]\geq 1-\varepsilon\] for the $\beta_0$ which minimizes $\E_{\alpha\sim\nu_0}\left[g(\alpha,\beta)\right]$.  Therefore we have
\begin{align*}
V_F+V_G&=\min_\nu\max_\beta \E_{\alpha\sim\nu}\left[f(\alpha,\beta)\right]+\max_\nu\min_\beta\E_{\alpha\sim\nu}\left[g(\alpha,\beta)\right]\\
&\geq \max_\beta\E_{\alpha\sim\nu_0}\left[f(\alpha,\beta)\right]+\min_\beta\E_{\alpha\sim\nu_0}\left[g(\alpha,\beta)\right]\\
&\geq\E_{\alpha\sim\nu_0}\left[f(\alpha,\beta_0)\right]+\E_{\alpha\sim\nu_0}\left[g(\alpha,\beta_0)\right]\\
&\geq 1-\varepsilon.\qedhere
\end{align*}
\end{proof}

Now fix $\mu_1$ and $\nu_1$ attaining the maxima in the expressions $\max_\mu\min_\alpha \E_{\beta\sim\mu}\left[f(\alpha,\beta)\right]$ and $\max_\nu\min_\beta \E_{\alpha\sim\nu}\left[g(\alpha,\beta)\right]$ respectively.  Consider the distributions 
\begin{align*}
\Delta_F&\coloneqq\left((B_i\vert_\beta\wedge G\vert_\beta)\setminus F,\Pi_i\vert_\beta\right)_{i\sim\Delta,\beta\sim\mu_1},\\
\Delta_G&\coloneqq \left((B_i\vert_\alpha\wedge F\vert_\alpha)\setminus G,\Pi_i\vert_\alpha\right)_{i\sim\Delta,\alpha\sim\nu_1}.
\end{align*}
These are random resolution refutations of $F$ and $G$ respectively.  For every $\alpha$,  we get,  using Lemma~\ref{lemm:trav},  that the success rate (i.e.\ one minus its dynamic error) of $\Delta_F$ on $\alpha$ is
\begin{align*}
\Pr_{i\sim\Delta,\beta\sim\mu_1}\left[\tau_{\scriptscriptstyle\Pi_i\vert_\beta}(\alpha)\in F\right]&\geq\Pr_{i\sim\Delta,\beta\sim\mu_1}\left[\tau_{\scriptscriptstyle\Pi_i}(\alpha,\beta)\in F\right]\\
&=\E_{\beta\sim\mu_1}\left[f(\alpha,\beta)\right]
\geq V_F,
\end{align*}
hence its overall success rate is at least $V_F$.  Similarly,  we get that the success rate of $\Delta_G$ is at least $V_G$.  Since $V_F+V_G\geq 1-\varepsilon$ by Lemma~\ref{lemm:main},  either $V_F$ or $V_G$ must be at least $(1-\varepsilon)/2$.
We have thus proved:
\begin{thm}\label{thrm:main}
If $F\wedge G$ has a random resolution refutation of size $s$ and dynamic error $\varepsilon$,  then either $F$ or $G$ has a random resolution refutation of size $s$ and dynamic error at most $(1+\varepsilon)/2$.
\end{thm}

Notice that for $\varepsilon<1$,  the condition $V_F\geq (1-\varepsilon)/2$ is an interpolant, that is,  it correctly decides which of $F$ and $G$ is unsatisfiable.  Unfortunately,  it is not a polynomial-time computable interpolant.  Let us also note that the fact that the success rate of the proof given by Theorem~\ref{thrm:main} drops by half,  does not really make the result weaker,  as we can always reduce the error by Proposition~\ref{prop:prop}.  Specifically,  Proposition~\ref{prop:prop} gives:
\begin{cor}
If $F\wedge G$ has a random resolution refutation of size $s$ and dynamic error $\varepsilon\in(0,1)$,  then either $F$ or $G$ has a random resolution refutation of dynamic error at most $\varepsilon$ and size $s^{O(1+\log(1/\varepsilon))}$.
\end{cor}
In particular,  for every $\varepsilon\in(0,1)$,  random resolution with dynamic error $\varepsilon$ has the feasible disjunction property.

\section{Concluding remarks}

We have identified a proof system, namely random resolution with dynamic error,  which has the feasible disjunction property,  yet the question of whether it also has the feasible interpolation property remains open.  This answers,  to the extent we are willing to believe that random resolution does not have feasible interpolation,  a question raised by Pudl\'{a}k~\cite{Pudl03},  who asks whether feasible disjunction holds for any of the systems that do not have feasible interpolation.  It should be noted that natural as it may be,  random resolution is not a traditional proof system; in particular it is not a proof system in the sense of Cook and Reckhow \cite{CookReck79},  as we cannot check in polynomial time if a given refutation is a valid random resolution refutation,  unless $\mathsf{P}=\mathsf{NP}$ (see \cite{PudlThap19}).  Hence Pudl\'{a}k's  question remains for Cook-Reckhow systems.

Whether random resolution has feasible interpolation is an important question for one more reason.  If random resolution had feasible interpolation,  then resolution over parities (see \cite{ItsySoko20}),  also known as Res($\oplus$) or R(LIN/$\mathbb{F}_2$),  would have feasible interpolation.  This follows from the fact that equality has small randomized communication protocols.  Res($\oplus$) is an important proof system,  right beyond the boundary of proof systems for which lower bounds can be proved.  Closer to the concerns of this paper,  it also lies at the boundary beyond which feasible interpolation is believed to fail.  In particular,  it can be seen that the interpolation problem for Res($\oplus$) can be reduced to the interpolation problem for Res(2),  which,  as we noted before,  is where,  in the bounded-depth Frege hierarchy,  feasible interpolation seems to fail.

It is worth noting that our main theorem,  Theorem~\ref{thrm:main},  is quite general.  Namely,  it holds for the random version of any proof system satisfying Lemma~\ref{lemm:trav},  and for example Frege systems (under an appropriate formalization) satisfy this lemma.  However,  dynamic error makes random versions of proof systems more general than resolution become trivial very quickly.  Specifically,  let $F=C_1\wedge\dotsb\wedge C_m$ be an unsatisfiable CNF.  Consider the pair $(\neg C_1,\Pi)$,  where in $\Pi$ we start from $\neg C_1$,  derive $\neg C_1\vee\dotsb\vee\neg C_m$,  and then use $m$ successive cuts on the clauses $C_1,\dotsc,C_m$ to derive the empty clause.  This is a depth-2 Frege refutation of $F$ of dynamic error 0 (but static error 1).

Considering that dynamic and static error behave so differently for random depth-2 Frege (this is a non-trivial system for static error,  see \cite{PudlThap19}),  we believe it is worth investigating how random resolution with dynamic error $\varepsilon$ compares to random resolution with static error $\varepsilon$.  A first observation is that the tree-like versions of the two systems coincide.  A second observation is that the lower bounds of \cite{PudlThap19} for DAG-like random resolution do not seem to immediately transfer to random resolution with dynamic error. Briefly,  the key ingredient in \cite{PudlThap19} for proving lower bounds is a ``fixing lemma'' showing that an appropriately chosen random partial assignment $\alpha$ either falsifies $B_i$ or admits no legal extension falsifying $B_i$,  with high probability.  Combined with the static error condition,  this yields an $\alpha$  no legal extension of which falsifies $B_i$,  and static error seems necessary for this step.  So an interesting problem is first to establish exponential lower bounds for random resolution with dynamic error,  say $1/2$,  possibly by showing a fixing lemma for traversals.  Secondly,  to see whether the two systems,  random resolution with static error $\varepsilon$ and random resolution with dynamic error $\varepsilon$,  are polynomially equivalent.

\section*{Acknowledgment}
I would like to thank Pavel Pudl\'{a}k and Neil Thapen for fruitful comments on an earlier draft.

\bibliographystyle{alphaurl}
\bibliography{refs}

@article{CookReck79,
    author = {Stephen Cook and Robert Reckhow},
    title = {The relative efficiency of propositional proof systems},
    journal = {Journal of Symbolic Logic},
    year = {1979},
    volume = {44},
    pages = {36--50}
}

@article{Boneetal00,
  author       = {Maria Luisa Bonet and
                  Toniann Pitassi and
                  Ran Raz},
  title        = {On Interpolation and Automatization for {F}rege Systems},
  journal      = {{SIAM} Journal on Computing},
  volume       = {29},
  pages        = {1939--1967},
  year         = {2000}
}

@article{KrajPudl98,
  author       = {Jan Kraj{\'i}{\v c}ek and
                  Pavel Pudl{\'{a}}k},
  title        = {Some Consequences of Cryptographical Conjectures for {$S^1_2$} and {EF}},
  journal      = {Information and Computation},
  volume       = {140},
  pages        = {82--94},
  year         = {1998}
}

@article{Pudl03,
  author       = {Pavel Pudl{\'{a}}k},
  title        = {On reducibility and symmetry of disjoint {NP} pairs},
  journal      = {Theoretical Computer Science},
  volume       = {295},
  pages        = {323--339},
  year         = {2003}
}

@article{Becketal14,
  author       = {Arnold Beckmann and
                  Pavel Pudl{\'{a}}k and
                  Neil Thapen},
  title        = {Parity Games and Propositional Proofs},
  journal      = {{ACM} Transactions on Computational Logic},
  volume       = {15},
  pages        = {17:1--17:30},
  year         = {2014}
}

@article{Kraj97,
  author       = {Jan Kraj{\'i}{\v c}ek},
  title={Interpolation theorems, lower bounds for proof systems, and independence results for bounded arithmetic},   
  journal={Journal of Symbolic Logic}, 
  volume={62},
  pages={457–486},
  year={1997}, 
}

@article{Pudl97,
  author       = {Pavel Pudl{\'{a}}k},
  title        = {Lower Bounds for Resolution and Cutting Plane Proofs and Monotone
                  Computations},
  journal      = {Journal of Symbolic Logic},
  volume       = {62},
  pages        = {981--998},
  year         = {1997}
}

@inproceedings{Hako20,
  author       = {Tuomas Hakoniemi},

  title        = {Feasible Interpolation for Polynomial Calculus and Sums-Of-Squares},
  booktitle    = {Proceedings of 47th International Colloquium on Automata, Languages, and Programming},
  volume       = {168},
  pages        = {63:1--63:14},
  year         = {2020}
}

@article{Kraj18,
  author       = {Jan Kraj{\'i}{\v c}ek},
  title        = {Randomized feasible interpolation and monotone circuits with a local oracle},
  journal      = {Journal of Mathematical Logic},
  volume       = {18},
  pages        = {1850012-1--27},
  year         = {2018}
}

@article{Bussetal14,
  author       = {Samuel Buss and
                  Leszek Ko{\l}odziejczyk and
                  Neil Thapen},
  title        = {Fragments of Approximate Counting},
  journal      = {Journal of Symbolic Logic},
  volume       = {79},
  pages        = {496--525},
  year         = {2014}
}

@article{PudlThap19,
  author       = {Pavel Pudl{\'{a}}k and
                  Neil Thapen},
  title        = {Random resolution refutations},
  journal      = {Computational Complexity},
  volume       = {28},
  pages        = {185--239},
  year         = {2019}
}

@article{ItsySoko20,
  author       = {Dmitry Itsykson and
                  Dmitry Sokolov},
  title        = {Resolution over linear equations modulo two},
  journal      = {Annals of Pure and Applied Logic},
  volume       = {171},
  year         = {2020}
}

@InProceedings{Rudi97,
author="Rudich, Steven",
title="Super-bits, demi-bits, and {NP}/qpoly-natural proofs",
booktitle="Proceedings of the Randomization and Approximation Techniques in Computer Science international
                  workshop",
year="1997",
pages="85--93"
}

@inproceedings{Garl24,
  author       = {Michal Garl{\'{\i}}k},
  title        = {Failure of Feasible Disjunction Property for k-{DNF} Resolution and
                  {NP}-Hardness of Automating It},
  booktitle    = {39th Computational Complexity Conference},
  volume       = {300},
  pages        = {33:1--33:23},
  year         = {2024}
}

@inproceedings{Hubaetal24,
  author       = {Pavel Hub{\'{a}}{\v c}ek and
                  Erfan Khaniki and
                  Neil Thapen},
  title        = {{TFNP} Intersections Through the Lens of Feasible Disjunction},
  booktitle    = {Proceedings of the 15th Innovations in Theoretical Computer Science Conference},
  volume       = {287},
  pages        = {63:1--63:24},
  year         = {2024}
}

@inproceedings{Lietal24,
  author       = {Yuhao Li and
                  William Pires and
                  Robert Robere},
  title        = {Intersection Classes in {TFNP} and Proof Complexity},
  booktitle    = {Proceedings of the 15th Innovations in Theoretical Computer Science Conference},
  volume       = {287},
  pages        = {74:1--74:22},
  year         = {2024}
}

@article{Kraj17,
    title      = {A feasible interpolation for random resolution},
    author     = {Jan Kraj{\'i}{\v c}ek},
    journal    = {Logical Methods in Computer Science},
    volume     = {13},
    year       = {2017}
}

\end{document}